\documentclass{article}
\usepackage[utf8]{inputenc}
\usepackage{fullpage}
\usepackage{amsmath,amsthm,amsfonts,amssymb}
\usepackage{tikz}
\usepackage{enumitem}
\usepackage{hyperref}
\usepackage[capitalize,nameinlink]{cleveref}
\usepackage{thmtools}
\usepackage{thm-restate}
\usepackage{hyperref}
\usepackage{multicol}
\usepackage[capitalize,nameinlink]{cleveref}

\usepackage[linesnumbered,boxed,noend,noline]{algorithm2e}
\usepackage{framed}
\usepackage{algpseudocode}

\SetKwFor{RepTimes}{repeat}{times}{end}

\usepackage{enumitem}
\usepackage{tcolorbox}
\makeatletter
\newcommand{\customlabel}[2]{%
   \protected@write \@auxout {}{\string \newlabel {#1}{{#2}{\thepage}{#2}{#1}{}} }%
   \hypertarget{#1}{#2}
}
\makeatother

\usepackage{cancel}
\crefname{algocf}{Algorithm}{Algorithms}
\Crefname{algocf}{Algorithm}{Algorithms}
\setenumerate[1]{align=left,label=\arabic*}
\setenumerate[2]{before=\stepcounter{enumi},label*=.\arabic*,leftmargin=0em,align=left}

\newcommand{\N}{\mathbb{N}}
\newcommand{\E}{\mathbb{E}}
\newcommand{\F}{\mathbb{F}}

\newcommand{\C}{\mathbb{C}}

\newcommand{\seq}{\subseteq}

\newcommand{\ip}[1]{\left\langle #1 \right\rangle}

\newcommand{\agr}{{\rm agr}}
\renewcommand{\sp}[1]{{\rm span}\{#1\}}

\newcommand{\abs}[1]{{\left|#1\right|}}

\renewcommand{\sp}{\mathsf{span}}
\newcommand{\ALG}{{\mathrm ALG}}

\newcommand{\Spec}{{\mathrm{Spec}}}

\theoremstyle{plain} 
\newtheorem{theorem}{Theorem}[section]
\newtheorem{lemma}[theorem]{Lemma}
\newtheorem{proposition}[theorem]{Proposition}
\newtheorem{claim}[theorem]{Claim}

\newtheorem{definition}[theorem]{Definition}

\newtheorem{remark}[theorem]{Remark}

\newcommand{\FormatAuthor}[3]{
\begin{tabular}{c}
#1 \\ {\small\texttt{#2}} \\ {\small #3}
\end{tabular}
}

\newcommand{\igor}[1]{\textcolor{blue}{\textbf{(Igor: #1})}}

\title{Matrix Multiplication Reductions}

\author{ 
\hspace*{-1.15cm}
\begin{tabular}[h!]{cccc}
\FormatAuthor{Ashish Gola}{ashish\_kg@sfu.ca}{Simon Fraser University} & 
\FormatAuthor{Igor Shinkar}{ishinkar@sfu.ca}{Simon Fraser University} & 
\FormatAuthor{Harsimran Singh}{harsimran\_singh\_3@sfu.ca}{Simon Fraser University}
\end{tabular}
} %

\begin{document}

\maketitle

\begin{abstract}
    In this paper we study a worst case to average case reduction for the problem of matrix multiplication over finite fields. Suppose we have an efficient \emph{average case} algorithm, that given two random matrices $A,B$ outputs a matrix that has a non-trivial correlation with their product $A \cdot B$. Can we transform it into a \emph{worst case} algorithm, that outputs the correct answer for all inputs without incurring a significant overhead in the running time? We present two results in this direction.

    \begin{description}
        \item[Two-sided error in the high agreement regime]
        We begin with a brief remark about a reduction for high agreement algorithms, i.e., an algorithm which agrees with the correct output on a large (say $>0.9$) fraction of entries, and show that the standard self-correction of linearity allows us to transform such algorithms into algorithms that work in worst case.

        \item[One-sided error in the low agreement regime]
        Focusing on average case algorithms with one-sided error, we show that over $\F_2$ there is a reduction that gets an $O(T)$ time \emph{average case} algorithm that given a random input $A,B$ outputs a matrix that agrees with $A \cdot B$ on at least $51\%$ of the entries (i.e., has only a slight advantage over the trivial algorithm), and transforms it into an $\widetilde{O}(T)$ time \emph{worst case} algorithm, that outputs the correct answer for \emph{all inputs} with high probability.
    \end{description}

\end{abstract}

\section{Introduction}

The problem of efficiently multiplying two matrices has been  extensively  studied for decades. Improving on the straightforward $O(n^3)$ time algorithm, Strassen's algorithm~\cite{Strassen1969} computes the product of two matrices in time $O(n^{\log_2 7} = n^{2.807})$, and it is perhaps the most widely used in practice. Since then, a long and exciting line of research (\cite{Pan-V.Ya}, \cite{BINI-et.al}, \cite{Schonhage}, \cite{Romani}, \cite{Strassen-86}, \cite{COPPERSMITH}, \cite{Stothers}, \cite{Williams-12}, \cite{Le-Gall}, \cite{AV21}) has led to a significant improvement of the value of the optimal exponent of the running time for matrix multiplication problem. The fastest algorithm known today is due to Duan, Wu, and  Zhou \cite{Duan-Wu-Zhou}, and its running time is $O(n^{2.371866})$.

Worst-case to average-case reductions serve as a means to convert algorithms that output correct answers on a fraction of inputs into algorithms with correct outputs on all possible inputs. These reductions can be viewed from two different perspectives. From the hardness point of view, they can be used to show that a problem maintains its hardness even in the average case. From the algorithmic side, they provide a framework for developing worst-case algorithms, by first designing weak algorithms with average case guarantees, and then transforming them into algorithms which work on all outputs.

In this paper, we study the following variant of a worst-case to average-case reduction for the matrix multiplication problem. Suppose we have an efficient algorithm that given two random matrices $A,B \in \F^{n \times n}$, computes a matrix $C \in \F^{n \times n}$ that agrees with the product $A \cdot B$ on a large fraction of the entries of the matrix.
Can we transform such an algorithm into one that computes $A \cdot B$ correctly for all entries of the output matrix without incurring a significant overhead in the running time? 

More formally, we define the \emph{agreement} between two matrices as the fraction of entries on which the two matrices agree.
\begin{definition}
Let $\F$ be a field, and let $A, B\in \F^{n \times n}$ be two matrices. We define agreement between $A$ and $B$, denoted by $\agr(A,B)$, as the fraction of entries $(i,j)$ on which $A_{i,j} = B_{i,j}$, i.e.,
\begin{equation*}
    \agr(A, B) = \frac{|\{(i,j):A_{i,j} = B_{i,j}\}|}{n^2}
    \enspace.
\end{equation*}
\end{definition}

Then, our goal can be stated as the task of transforming an algorithm that on a random input $A,B$ outputs a matrix $C$ such that $\agr(C,AB) \geq \alpha$ for some parameter $\alpha \in [0,1]$ into an algorithm that solves the matrix multiplication problem \emph{correctly on all inputs}.

We present two results in this direction. Both results consider the matrix multiplication problem over finite fields. 

\paragraph{High agreement regime with two-sided error:}
We show that any algorithm that solves the matrix multiplication problem correctly on a high fraction of the coordinates, can be converted into a worst case algorithm. Specifically, we prove the following theorem.

\begin{restatable}{theorem}{TWOSIDED}\label{thm:highagreement_mm}
    Fix a finite field $\F$.
    Let $\alpha \in (0, 1/8)$.
    Let $\ALG$ be an algorithm that gets as input two matrices $A,B \in \F^{n \times n}$, runs in time $T(n)$, and outputs a matrix $C \in \F^{n \times n}$. Suppose that
    \begin{equation*}
        \E_{A,B  \in \F^{n \times n}}[\agr(ALG(A,B), A \cdot B)] > 1-\alpha
        \enspace.
    \end{equation*}
    Then, there is an algorithm $\ALG^*$ that gets as input two matrices $A,B \in \F^{n \times n}$, runs in time $O(T(n) \cdot \log(n))$, and outputs a matrix $C \in \F^{n \times n}$ such that \emph{for all} $A,B$ it holds that
    \begin{equation*}
        \Pr[\ALG^*(A,B) =  A \cdot B] > 1-1/n
        \enspace,
    \end{equation*}
    where the randomness is only over the internal coins of $\ALG^*$.
\end{restatable}

The proof of this result relies on rather standard ideas, and essentially uses the self-correction of linear functions~\cite{BlumLR90}.

\paragraph{Low agreement with one-sided error:}
For this result, we restrict our discussion to the finite field $\F_2$.
Note that it is trivial to design an $O(n^2)$ time algorithm such that 
$\E_{A,B  \in \F_2^{n \times n}}[\agr(ALG(A,B), A \cdot B)] \geq 1/2$.
Indeed, the algorithm can simply output 0 in all entries irrespective of the input. Alternatively, the algorithm can output a random 0/1 matrix. Hence, it is natural to ask whether it is possible to obtain a better-than-1/2 algorithm for the matrix multiplication over $\F_2$.

Below we show that in the special case of \emph{one-sided error} approximation, any better-than-1/2 approximation $O(T)$ time algorithm can be transformed into a worst case algorithm with running time $\widetilde{O}(T)$. Formally, we prove the following theorem.

\begin{restatable}{theorem}{ONESIDED}\label{thm:mm-alg-reduction-one-sided}
    Let $\ALG$ be an algorithm that gets input two matrices $A,B \in \F_2^{n \times n}$, runs in time $T(n)$, and outputs a matrix $C \in \F_2^{n \times n}$. Let $\delta > 0$, and suppose that
    \begin{itemize}
        \item
        $\E_{A,B  \in \F^{n \times n}}[\agr(ALG(A,B), A \cdot B)] \geq 1/2+\delta$.
        \item\label{item:one-sided-error} If $(AB)_{i,j} = 0$, then $\ALG(A,B)_{i,j} = 0$.
    \end{itemize}
Then, there is an algorithm $\ALG^*$ that gets as input two matrices $A,B \in \F^{n \times n}$, runs in time $\widetilde{O}(T(n))$, and outputs a matrix $C \in \F^{n \times n}$ such that for all $A,B$ it holds that
    \begin{equation*}
        \Pr[\ALG^*(A,B) =  A \cdot B] > 1 - 1/n,
    \end{equation*}
    where the randomness is only over the internal coins of $\ALG^*$.
\end{restatable}

\begin{remark}
Below we make several comments about \cref{thm:mm-alg-reduction-one-sided}.

\begin{enumerate}
    \item Note that the conditions of the theorem can be written equivalently as follows.
    \begin{itemize}
        \item
    $\Pr_{\substack{A,B  \in \F_2^{n \times n} \\ i,j \in [n]}}[ALG(A,B)_{i,j} = 1] \geq \delta$.
    \item If $(AB)_{i,j} = 0$, then $\ALG(A,B)_{i,j} = 0$.
    \end{itemize}
    \item The notion of algorithms with one-sided error is typically studied in the context of randomized algorithms, e.g., relating to the classes $\mathcal{RP}$ (and $co\mathcal{RP}$), where the guarantee is that for \emph{every NO input} the algorithm outputs the correct answer with probability 1, and for \emph{every YES input} it is correct with probability at least 2/3. The error model in \cref{thm:mm-alg-reduction-one-sided} is different, as we consider algorithms that are correct on \emph{random inputs} on all output 0-bits, and on at least some $\alpha$-fraction of 1-bits.
    
    \item We remark that the standard methods of self-correcting linear functions work in the high agreement regime, but fail when the average case guarantee is low. We apply the techniques from additive combinatorics developed in~\cite{AsadiGGS22}, particularly a version of the probabilistic Bogolyubov-Ruzsa Lemma, to perform a self-correction procedure which helps in this regime.

    \item Our proof of \cref{thm:mm-alg-reduction-one-sided} assumes that $\ALG$ is deterministic. It is rather straightforward to extend the proof and allow it to be randomized, by appropriately modifying the sets of good inputs ($X_{i,j}$ and $Y_{i,j}^A$) to account for the randomness of the algorithm.
\end{enumerate}
\end{remark}

\subsection{Related work}
The study of average-case complexity began with Levin's work \cite{Levin-86}, followed by subsequent works like \cite{Ben-David_et_al_92}. A substantial body of research (e.g., \cite{Impagliazzo-Levin90}, \cite{Impagliazzo2011} and related references) identified numerous barriers in formulating worst-case to average-case reductions for NP-complete problems. For a comprehensive overview of this subject, see the classical surveys by Impagliazzo \cite{Impagliazzo1995}, Bogdanov and Trevisan \cite{Bogdanov-Trevisan-2006} and Goldreich \cite{Goldreich-2011}.

Asadi et al.~\cite{AsadiGGS22, AsadiGGSS24} presented a new framework for carrying out efficient worst-case to average case reductions for various fundamental problems. Particularly, for the problem of matrix multiplication, they proved that if there exists an $O(T(n))$ time algorithm $M$ for matrix multiplication which computes the correct output on an $\epsilon$ fraction of inputs, then there exists a randomized algorithm $M'$ which computes the correct output on all inputs, running in time $O(\exp(O(\log^5(1/\epsilon)) \cdot T(n))$. The proof relied on additive combinatorial techniques and used the probabilistic quasi-polynomial Bogolyubov-Ruzsa Lemma. 

Hirahara and Shimizu ~\cite{Hirahara-Shimizu} improved the $\exp(O(\log^5(1/\epsilon)))$ overhead to an $\widetilde{O}(1/\epsilon)$ factor. Their idea involved dividing the output matrix into smaller blocks and using the Direct-Product Theorem in a black-box manner.

The aforementioned papers assume that we have access to an algorithm which gives a fully correct output on some fraction of the inputs, i.e., for these inputs \emph{all entries} in the output matrix are correct. The setting presented in this paper, where the output of the given algorithm is not fully correct, seems to differ significantly from the works mentioned above.
In particular, we do not see how to apply the Direct-Product theorem to our setting of the problem.

\subsection{Open problems}

We mention the following two problems that are left open in this work.
\paragraph{Low agreement with two-sided error:}Is it possible to transform a two-sided error algorithm over $\F_2$ with a low agreement guarantee into a worst case algorithm.
That is, given an $O(T(n))$ time algorithm $ALG$ with the guarantee
$\E_{A,B  \in \F^{n \times n}}[\agr(ALG(A,B), A \cdot B)] > 1/2+\delta$,
can we convert it into an algorithm that correctly outputs the correct answer on all inputs and has running time $\widetilde{O}(T(n))$?
\paragraph{Generalizing over finite fields:} Extend \cref{thm:mm-alg-reduction-one-sided} in a meaningful way to work over any finite field.

\section{Preliminaries}

For a positive integer $n$ we define $[n] = \{0,1,\dots,n-1\}$.
We index the coordinates of our matrices starting from 0 rather than 1, which is typically more standard. We refer to the element in the row $i$ and column $j$ of the matrix $A$ as $A_{i,j}$.

We define a notion of \emph{row-shift} (or \emph{row-rotation}) and \emph{column-shift} as follows.

\begin{definition}
    Given a matrix $A \in \F^{n \times m}$, $0 \leq \pi \leq n-1$, and $0 \leq \sigma \leq m-1$,
    define $A^{\pi,\sigma}$ to be the matrix obtained from $A$ by cyclically rotating all its rows downwards by $\pi$ units and all its columns rightwards by $\sigma$ units, that is,
    \begin{equation*}
        (A^{\pi,\sigma})_{i,j} = A_{(i - \pi) \bmod n,  (j-\sigma) \bmod m}
        \enspace
    \end{equation*}
\end{definition}

The following proposition is immediate from the definition above.

\begin{proposition}
For any $A,B,C \in \F^{n \times n}$ and any $\pi, \sigma$ we have $AB = C$ if and only if $A^{\pi,0} \cdot B^{0,\sigma} = C^{\pi,\sigma}$.
\end{proposition}

\subsection{Additive Combinatorics Tools}\label{section:additive-combinatorics}
We now present the additive combinatorics toolkit which will be useful in the worst-case to average-case reduction for the low agreement regime with one-sided error. 

For a set $A \subseteq \F_2^n$, let $1_A \colon \F_2^n \to \{0,1\}$ denote the indicator function of A. The Fourier expansion of a function $f \colon \F_2^n \to \C$ is given by $f(x) = \sum_{r \in \F_2^n} \hat{f}(r) \cdot \chi_r(x)$,
where $\chi_r(x) = (-1)^{\ip{x,r}}$, and the Fourier coefficients of $f$ are defined as $\hat{f}(r) = \ip{f,\chi_r} = \E_x[f(x) \cdot \chi_r(x)]$.
\emph{Parseval's identity} says that $\sum_{r \in {\F_2}^n} {\widehat{1_A}(r)}^2 = \ip{1_A,1_A}  = \alpha$, where $\alpha$ is the density of $A$.

Define $\Spec_{\gamma}(A) = \{r \in \F_2^n : \abs{\widehat{1_A}(r)} \geq \gamma\}$. Below we state Chang's lemma, which describes a certain structure of $\Spec_{\gamma}(A)$.

\begin{lemma}[Chang's Theorem \cite{Chang02}]\label{lemma:chang}
	Let $A \seq \F^n$ be a set of size $\abs{A} = \alpha \cdot \abs{\F}^n$, and let $\gamma>0$. Then
	\begin{equation*}
	      \dim(\sp(\Spec_{\gamma\alpha}(A))) \leq O \left(\frac{\log(1/\alpha)}{\gamma^2} \right)
	      \enspace.
	\end{equation*}
\end{lemma}

Recall that the subset sum of two sets $A$ and $B$ is defined as $A+B =\{a+b : a\in A \;, b \in B\}$.
Analogously, we define $tA = A+A+\dots+A$ ($t$ times) as $tA = \{a_1+a_2\dots+a_t : a_1,a_2,\dots,a_t \in A\}$.
The following lemma says that for an arbitrary set $A \seq \F^n$, the sumset $tA$ contains a large affine subspace.

\begin{lemma}[Probabilistic quasi-polynomial Bogolyubov-Ruzsa lemma]\label{lemma:bogolyubov-quasipoly}
    Let $A \seq \F_2^n$ be a set of size $\abs{A} = \alpha \cdot 2^n$, for some $\alpha \in (0,1]$, and let $t \geq 3$ be an integer. Then, $tA$ contains an \emph{affine} subspace $V \seq \F_2^n$ of dimension $\dim(V) \geq n - O(\log(1/\alpha))$ such that for all $v \in V$ it holds that
    \begin{equation*}
        \Pr_{a_1,a_2,..,a_{t-1} \in \F_2^n}[a_1, a_2, a_3,..,a_t \in A] \geq 
        {\alpha}^t\left( 1 + \frac{1}{2^{t-2}} \right)  - \frac{{\alpha}^{t-1}}{2^{t-2}}
        \enspace,
    \end{equation*}
    where $a_t = v-a_1-a_2-..-a_{t-1}$.

    In particular, if $t > \log_2(1/\alpha)+2$, then $tA$ contains an \emph{affine} subspace $V \seq \F_2^n$ of dimension $\dim(V) \geq n - k$, for $k = O(\log(1/\alpha))$, such that for all $v \in V$ it holds that
    \begin{equation*}
        \Pr_{\substack{a_1,a_2,..,a_t \in \F_2^n \\ v = \sum_{i=1}^t a_i}}[a_1, a_2, a_3,..,a_t \in A] \geq (\alpha/2)^t
        \enspace.
    \end{equation*}
    \end{lemma}

Below we prove \cref{lemma:bogolyubov-quasipoly} only for odd values of $t$, which is slightly more complicated than the case of even $t$. After the proof, we remark how to modify the proof to work for even $t$'s.

\begin{proof} 
Let $A \seq \F_2^n$ be a set of size $\abs{A} = \alpha \cdot \abs{\F}^n$, for some $\alpha \in (0,1]$. 
Consider the set
\begin{equation*}
    R = \Spec_{\alpha/2}\setminus \{0\} = \{ r \in F_2^{n}\setminus\{0\} : \abs{\widehat{1_A}(r)}  > \frac{\alpha}{2} \}
\end{equation*}

Next we define an affine subspace $V = \{v \in \F_2^n : \ip{v,r} = s_r \ \forall r \in R\}$ for some $s_r \in \{0,1\}$ to be defined later, and claim that $V$ satisfies the conclusions of \cref{lemma:bogolyubov-quasipoly}.
We will need the following two claims.

\begin{claim}\label{claim:affine}
    For all $r \in R$ there exists $s_r \in \{0,1\}$ such that
    (1) $\sum_{r \in R} \widehat{1_A}(r)^t \cdot (-1)^{s_r} \geq 0$
    and (2) if $r^* \in R$ is a linear combination $r^* = \sum_{r \in R} c_r \cdot r$ of vectors in $R$ (with $c_r \in \F_2$),
    then $s_{r^*} = \sum_{r \in R} c_r \cdot s_r \pmod 2$.
\end{claim}

\begin{proof}
    Let $R'$ be a maximal subset of $R$ of linearly independent vectors.
    Choose $s_{r'} \in \{0,1\}$ independently with probability 0.5 each  for every $r' \in R'$.
    Now any $r \in R \setminus R'$, can be expressed as a linear combination $r = \sum_{r'} c_{r'} \cdot r'$ of vectors in $R'$ with $c_{r'} \in \{0,1\}$
    define $s_r = \sum_{r'} c_{r'} \cdot s_{r'}$. It is immediate to verify that condition (2) is satisfied.

    In order to satisfy condition (1) note that by linearity of expectation $\E[\sum_{r \in R} \widehat{1_A}(r)^t \cdot (-1)^{s_r}] = 0$, and hence there exists a choice of $(s_r)_{r \in R}$ such that 
    $\sum_{r \in R} \widehat{1_A}(r)^t \cdot (-1)^{s_r} \geq 0$, as required.
\end{proof}

\begin{claim}\label{claim:bound-small}
    We have $\sum_{r \not\in R, r \neq 0} {\abs{\widehat{1_A}(r)}}^t \leq (\alpha/2)^{t-2}(\alpha - {\alpha}^2 )$.
\end{claim}
\begin{proof}
For $t\geq 3$, it holds that
\begin{eqnarray*}
	\sum_{r \not\in R, r \neq 0} {\abs{\widehat{1_A}(r)}}^t 
	\leq  \max_{r \not\in R, r \neq 0}{\abs{\widehat{1_A}(r)}}^{t-2}\sum_{r \not\in R, r \neq 0} {\abs{\widehat{1_A}(r)}}^2 \\
	\leq  (\alpha/2)^{t-2}\sum_{r \in {\F_2}^n \setminus \{0\}} {\widehat{1_A}(r)}^2 \\
	<  (\alpha/2)^{t-2}(\alpha - {\alpha}^2 )
    \enspace.
\end{eqnarray*}
\end{proof}

Define an affine subspace $V = \{v \in \F_2^n : \ip{v,r} = s_r \ \forall r \in R\}$,
where $s_r \in \{0,1\}$ is from \cref{claim:affine}.
Note that if the vectors in $R$ are linearly dependent, then
the second condition of \cref{claim:affine} guarantees that we can define $V = \{v \in \F^n : \ip{v,r'} = s_{r'} \ \forall r' \in R'\}$ for a maximal set $R' \subset R$ of linearly independent vectors in $R$, and the remaining constraints will be satisfied by linearity.
Then, according to \cref{lemma:chang} we have
\begin{equation*}
    \dim(V) \geq n - O \left(\log(1/\alpha) \right)
	\enspace.
\end{equation*}

Using the two claims above, and noting that $\Pr_{a_1,a_2,..,a_{t-1} \in \F^n}[a_1, a_2, a_3,..,a_t \in A] = 1_A*1_A*..*1_A(v)$ ($t$ times), for any $v \in V$ we have
\begin{eqnarray*}
	&\Pr_{a_1,a_2,..,a_{t-1} \in \F^n}[a_1, a_2, a_3,..,a_t \in A] &= 1_A*1_A*..*1_A(v) \\
    &&= \sum_{r \in \F^n} {\widehat{1_A}(r)}^t \chi_r(v) \\
	&&= {\widehat{1_A}(0)}^t + \sum_{r \in R} {\widehat{1_A}(r)}^t \chi_r(v) + \sum_{r \not\in R, r \neq 0} {\widehat{1_A}(r)}^t \chi_r(v)\\
	&& \geq {\alpha}^t + \sum_{r \in R} {\widehat{1_A}(r)}^t \cdot (-1)^{s_r} - (\alpha/2)^{t-2}(\alpha - {\alpha}^2) \\
	&& \geq {\alpha}^t + 0 - (\alpha/2)^{t-2}(\alpha - {\alpha}^2) \\
	&& = {\alpha}^t\left( 1 + \frac{1}{2^{t-2}} \right)  - \frac{{\alpha}^{t-1}}{2^{t-2}}
    \enspace.
 \end{eqnarray*}
 In particular, for $t > \log_2(1/\alpha)+2$, we have
\begin{eqnarray*}
	&\Pr_{a_1,a_2,..,a_{t-1} \in \F^n}[a_1, a_2, a_3,..,a_t \in A] & \geq {\alpha}^t\left( 1 + \frac{1}{2^{t-2}} \right)  - \frac{{\alpha}^{t-1}}{2^{t-2}} \\
	&& \geq {\alpha}^t\left( 1 + \frac{1}{2^{t-2}}\right)  - {\alpha}^t \\
	&& \geq (\alpha/2)^t
    \enspace,
\end{eqnarray*}
as required.
\end{proof}

\begin{remark}
    For even values of $t$ the lemma is slightly easier. Specifically,
    since ${\widehat{1_A}(r)}^t$ is always non-negative, we can take $s_r = 0$ in \cref{claim:affine}, and the rest of the proof works the same.
\end{remark}

\section{High Agreement with Two-Sided Error}

In this section, we prove \cref{thm:highagreement_mm}. Specifically, we show that if there exists an algorithm which, given two matrices $A, B \in \F^{n \times n}$, runs in time $T(n)$ and correctly computes their product on a large fraction of all entries of output on average, then there exists another algorithm that runs in  $\widetilde{O}(T(n))$ time and correctly computes their product on all entries of output. The proof essentially uses the self-correction of linearity~\cite{BlumLR90}.

\TWOSIDED

\begin{proof}
    Given the algorithm $\ALG$ as in the assumption of the theorem, we design $\ALG^*$ as follows.

	\begin{algorithm}
	\caption{Approximation for High Agreement Matrix Multiplication Algorithms}\label{alg:highagreement-mm}
        \KwIn{$A,B \in\F^{n \times n}$, $ALG$}
        \KwOut{$A \cdot B$}

        Let $k = O(\log(n))$
        
        \For{$r = 0$ to $k$}{
		  Generate two random matrices $R,S\in\F^{n \times n}$

		  Select two random variables $\pi,\sigma \in [n]$ independently

		  $M = \ALG((A+R)^{\pi,0},(B+S)^{0,\sigma}) - \ALG(R^{\pi,0},(B+S)^{0,\sigma}) - \ALG((A+R)^{\pi,0}, S^{0,\sigma}) + \ALG(R^{\pi,0},S^{0,\sigma})$

            Let $C_r = M^{n-\pi, n-\sigma}$
		}
		Define the matrix $C \in \F^{n \times n}$ by taking the majority vote of all $C_r$ in each coordinate.
	\end{algorithm}
    \paragraph{Correctness:}
    Consider an entry $(i,j)$ in the output matrix $A \cdot B$.
    In each iteration we call $\ALG$ four times, and in each of the calls the input is distributed uniformly in $\F^{n \times n}$. Furthermore, since $\pi, \sigma$ are chosen uniformly, it follows that $(i+\pi, j+\sigma)$ are distributed uniformly.
    Therefore, the probability that that all the four calls of $\ALG$ produce the correct answer in this entry is at least $1 - 4\alpha$.
    Therefore, for each repetition $r$, we have
    \begin{equation*}
        \Pr[(C_r)_{i,j} = (A \cdot B)_{i,j}] \geq 1-4\alpha
        \enspace.
    \end{equation*}
    By Chernoff bound, the probably that the majority vote of the $k$ repetition will produce an incorrect answer is upper bounded by
	\begin{equation*}
        \Pr[(C_r)_{i,j} = (A \cdot B)_{i,j}] \geq \exp\left(-\Omega((1-4\alpha - 1/2) \cdot k)\right) < 1/n^3
        \enspace,
	\end{equation*}
    Here we make the assumption that $\alpha$ is bounded below $1/8$.

    Hence, the probability that a particular entry $(i,j)$ is incorrect after $k$ iterations is at most $n^{-3}$. By union bound over all entries, the probability that at least one entry is incorrect in the output matrix is at most $n^2 \cdot n^{-3}  = 1/n$.

    \paragraph{Running time:}
	The total running time is dominated by $O(\log(n))$ invocations of $\ALG$, and hence, the runtime of $\ALG^*$ is $O(T(n) \cdot \log(n))$.
\end{proof}

\section{Low Agreement with One-Sided Error}

In this section we prove \cref{thm:mm-alg-reduction-one-sided}.
We restate the theorem here for convenience.

\ONESIDED*

Before proving the theorem, we need some definitions. We start by defining the notion of a {\it good} coordinate. We say a coordinate $(i,j) \in [n] \times [n]$  is {\it good}, if $ALG$ returns $1$ at the entry $(i,j)$ for more than ${\delta}/{2}$ fraction of possible inputs.

    \begin{definition}\label{def:good-coordinates}
    Denote by $G$ the set of \emph{good} coordinates, defined as
    \begin{equation*}
        G = \{(i,j) \in [n] \times [n] : \Pr_{A,B  \in \F_2^{n \times n}}[ALG(A,B)_{i,j} = 1] > \delta/2 \}
        \enspace.
    \end{equation*}
    \end{definition}

    The following claim is immediate from the definition and the assumptions of the theorem.
    \begin{claim}\label{prob-good-coordinate}
        $|G| \geq (\delta/2) \cdot n^2$.
    \end{claim}

    \begin{proof}
    Let $p_{i,j} = \Pr_{\substack{A,B  \in \F_2^{n \times n}}}[ALG(A,B)_{i,j} = 1]$. Note that by the assumptions of \cref{thm:mm-alg-reduction-one-sided}, we have $\E_{i,j}[p_{i,j}] \geq \delta$. Note that
    \begin{equation*}
        \delta \leq \E_{i,j \in [n] \times [n]}[p_{i,j}]
        \leq \Pr_{i,j}[(i,j) \in G] \cdot 1 + \Pr_{i,j}[(i,j) \not\in G]  \cdot (\delta/2)
        \leq \Pr_{i,j}[(i,j) \in G] \cdot 1 + 1 \cdot (\delta/2)
        \enspace,
    \end{equation*}
    and hence $\Pr[(i,j) \in G] \geq \delta/2$, as required. 
    \end{proof}
    
	Next, we define the notion of good input matrices with respect to a good coordinate. 
 
    \begin{definition}\label{defgood-matrices}
        For a coordinate (i,j), define $X_{i,j}$ as follows.    
        \begin{equation*}
            X_{i,j} = \{A: \Pr_{B}[\ALG(A,B)_{i,j}= 1]\geq \delta/4\}
            \enspace.
        \end{equation*}
        Given a coordinate $(i,j)$ and a matrix $A$, define $Y_{i,j}^A$ to be the set of matrices $B$ for which $\ALG$ returns 1 at  the entry $(i,j)$. That is,
        \begin{equation*}
            Y_{i,j}^A = \{B: \ALG(A,B)_{i,j}=1 \}
            \enspace.
        \end{equation*}    
    \end{definition}
    
    We make the following claim about the densities of $X_{i,j}$ and $Y_{i,j}^A$.
    \begin{claim}\label{claim-density}
    For any $(i,j) \in G$ it holds that $\Pr_{A \in \F^{n \times n}}[A \in X_{i,j}] \geq \delta/4$.
    Furthermore, if $A \in X_{i,j}$ then
    $\Pr_{B \in \F^{n \times n}}[B \in Y_{i,j}^A] \geq \delta/4$.
    \end{claim}
   
    \begin{proof}
    Fix a good coordinate $(i,j) \in G$, and for each $A \in \F^{n \times n}$ let $p_A = \Pr_{B \in \F^{n \times n}}[\ALG(A,B)_{(i,j)}=1]$. From the definition of $G$ we have $\E_A[p_A] \geq \delta /2$.
    \begin{equation*}
        \delta/2 \leq \E_A[p_A] = \Pr_{A \in \F^{n \times n}}[A \in X_{i,j}] \cdot 1 + \Pr_{A \in \F^{n \times n}}[A \not\in X_{i,j}] \cdot \delta/4 \leq \Pr_{A \in \F^{n \times n}}[A \in X_{i,j}] + \delta/4
        \enspace,
    \end{equation*}
    and hence, $\Pr[A \in X_{i,j}] \geq \delta/4$.

    The furthermore part is by definition of $X_{i,j}$.
    \end{proof}

    \begin{definition}\label{def:low_rank_matrix}
       Denote by $L^k \in \F^{n \times n}$ a random matrix of rank at most $k$, constructed by sampling the first $k$ columns independently uniformly at random from $\F^n$, and then taking the remaining $n-k$ columns to be uniformly random linear combinations of the first $k$ vectors.
    \end{definition}
    
    The following lemma is from \cite{AsadiGGS22}. It shows that 
    if $L^{2k}_A$ is a random matrix of rank at most $2k$ sampled as in \cref{def:low_rank_matrix},
    then $M_A = A - (L^{2k}_A)$ belongs to any subspace of matrices of co-dimension $k$ with a non-negligible probability. We provide the proof of the lemma here for completeness.
    
    \begin{lemma}[Lemma 4.8 from \cite{AsadiGGS22}]\label{lemma:extended-matrix}
        Fix a matrix $A \in \F_2^{n \times n}$, let $k$ be a parameter, and let $\ell \geq 2k$.
        Let $L^{\ell}_A$ be a random matrix of rank at most $\ell$ sampled as in \cref{def:low_rank_matrix},
        and let $M_A = A - (L^{\ell}_A)$.
        Then, for any subspace $V \seq \F^{n \times n}$ of $\dim(V) \geq n^2-k$ it holds that
        \begin{equation*}
            \Pr[M_A \in V] \geq \frac{1}{2\abs{\F}^k}\enspace.
        \end{equation*}
    \end{lemma}
    
    \begin{proof}
        Since $V$ has co-dimension at most $k$, a matrix in $V$ must be orthogonal to all the basis vectors of its orthogonal complement. Since there are up to $k$ such basis vectors, the membership condition of $M_A$ in $V$ can be written down in the form of $k$ linear constraints. Viewing $M_A$ as a vector in $\mathbb{F}^{n^2}$, we can write the $k$ linear constraints on the elements of the matrix $M^{2k}_A$ in the form
        \begin{equation*}
            \alpha_1 \cdot (M_A)_{i_1,j_1} + \alpha_2 \cdot (M_A)_{i_2,j_2} + \cdots + \alpha_r (M_A)_{i_t,j_t} = 0
            \enspace.
        \end{equation*}
Here, $\alpha_i$'s are constants and $t \in [n^2]$ is the number of elements upon which the constraints depend. Writing $M_A$ as $m \in \mathbb{F}^{n^2}$, we can represent these linear constraints as a system of equations of the form $G \cdot m=0$, where $G$ is a $k \times n^2$ matrix. Now, we perform Gaussian elimination on $G$, which gives us a matrix $G'$, where each row has a 1 entry such that all the other entries in the column containing this 1 are 0. We refer to such 1's as \emph{leading 1's}. That is, by permuting the columns of $G'$, we may think of it as being of the form $G' = [I_k | G^*]$. 

Consider the set of $k$ coordinates of $m$ corresponding to the $k$ leading $1$s in $G'$, one from each row. These $k$ coordinates of $m$ in turn correspond to $k$ pairs of coordinates $\{(i_1,j_1),(i_2,j_2) \ldots (i_k,j_k)\}$ in $M_A$. These $k$ pairs of coordinates in $M_A$ can belong to at most $k$ rows in  $M_A$.  We now bound the probability of none of these $k$ rows in $L_A^{\ell}$ being a linear combination of the other rows. Let us denote this event as $\Omega$. Then 
\begin{align*}
    \Pr[\Omega] &= \left(1- \frac{1}{2^\ell}\right) \left(1- \frac{2}{2^\ell}\right)  \left(1- \frac{4}{2^\ell}\right) \cdots \left(1- \frac{2^{k-1}}{2^\ell}\right) \\
    &\geq \left(1- \frac{2^{k-1}}{2^\ell}\right)^{k} \\
    &\geq \left(1- \frac{1}{2^{k+1}}\right)^{k} \\
    &\geq 1- \frac{k}{2^{k+1}} \geq \frac{1}{2}
\end{align*}
If $\Omega$ happens, then we get a coordinate $(i_r,j_r)$ in $M_A$ corresponding to the $r^{th}$ linear constraint, for all $r \in [k]$, such that no other constraint depends upon it (as it corresponds to a leading $1$) and it is chosen uniformly at random (since the rows containing these coordinates are linearly independent). Therefore, the probability that this random value satisfies the $i^{th}$ constraint is $1/\lvert\mathbb{F}\rvert$. To see this, assume that the values of all other coordinates involved in the $i^{th}$ constraint are fixed, then we are left with only one choice for the value of the coordinate $(c_i,c_i')$ which satisfies the constraint. Therefore, we have 
\begin{align*}
    \Pr[M_A \in V] &= \Pr[\text{All $k$ linear constraints are satisfied}] \\
    &= \Pr[\Omega] \cdot \frac{1}{\mathbb{\lvert F \rvert}^k} \geq \frac{1}{2\mathbb{\lvert F \rvert}^k}
    \enspace.
\end{align*}
This completes the proof of \cref{lemma:extended-matrix}.
\end{proof}

\subsection{Computing the good coordinates}\label{sec:good-coordinates}

Next, we start describing the reduction guaranteed by \cref{thm:mm-alg-reduction-one-sided}.
As a first step we design \cref{alg:forgoodentries}, that gets two matrices $A,B$ and outputs a matrix $C$ with values in $\F_2 \cup \{\bot\}$, satisfying the following guarantees.
\begin{enumerate}
    \item If $C_{i,j} \neq \bot$, then $C_{i,j}$ contains the correct values, i.e., $C_{i^*,j^*} = (A \cdot B)_{i^*,j^*}$.
    \item For any good coordinate $(i^*,j^*) \in G$ we have $\Pr[C_{i^*,j^*} = (A \cdot B)_{i^*,j^*}] \geq \delta_0$, where $\delta_0$ is some constant that depends only on $\delta$. That is, with non-negligible probability $C_{i^*,j^*}$ contains the correct answer, and not $\bot$.

\end{enumerate}
Then, in \cref{sec:proof-one-sided} we use \cref{alg:forgoodentries} as a subroutine, in order to compute the entire matrix $A \cdot B$ correctly.

    \begin{algorithm}
	\caption{Approximating good coordinates for one-sided error algorithms}\label{alg:forgoodentries}
    \KwIn{$A, B \in\F_2^{n \times n}$, $\ALG$}
    \KwOut{An $n \times n$ matrix $C$ with values $\F_2 \cup \{\bot\}$}

        Let $t > \log(4/\delta)+2$

        Let $k = O(\log(1/{\delta}))$ from the ``in particular'' part of \cref{lemma:bogolyubov-quasipoly} with $\alpha = \delta/4$ and $t$ chosen above

        Sample two random matrices $L_A^{2k}$ and $L_B^{2tk}$ of rank at most $2k$ and $2tk$ respectively, as in \cref{def:low_rank_matrix}
            
        Define $M_A = A - L_A^{2k}$ and $M_B = B - L_B^{2tk}$

        Let $C$ be the $n \times n$ matrix initialized with all $\bot$

        Sample $t-1$ random matrices $R_1, R_2,..,R_{t-1} \in \F^{n \times n}_2$ and set $R_t = M_A - (R_1 + R_2 +..+R_{t-1})$

        \For{$r = 1,\dots,t$} {
            Sample $t-1$ random matrices $S^{(r)}_1, S^{(r)}_2,.., S^{(r)}_{t-1} \in \F^{n \times n}_2$ and set $S^{(r)}_t = M_B - (S^{(r)}_1 + S^{(r)}_2 +.. +S^{(r)}_{t-1})$
        }
        
        \For{$r,s = 1,\dots,t$} {
            Compute $\ALG(R_r, S^{(r)}_s)$     
        }
        
        \For{$(i,j) \in [n] \times [n]$} {
            \If{$\ALG(R_r, S^{(r)}_s)_{i,j}=1$ for all $r,s \in \{1,\dots,t\}$} {
                Set $C_{i,j} = \sum_{r,s} \ALG(R_r, S^{(r)}_s)_{i,j} \pmod 2$
            }
        }

        \Return $C + M_A \cdot L_B^{2tk} + L_A^{2k} \cdot M_B + L_A^{2k} \cdot L_B^{2tk}$
        // if $C_{i,j} = \bot$, then we return $\bot$ in the coordinate $(i,j)$
    \end{algorithm}

    In lines 3-4 we decompose $A = L_A^{2k} + M_A$, and $B = L_B^{2tk} + M_B$ with the intention of computing $A \cdot  B$ by writing
    \begin{align*}
        AB &= (M_A + L_A^{2k})\cdot (M_B +L_B^{2tk}) \\
        &= M_A \cdot M_B + M_A \cdot L_B^{2tk} + L_A^{2k} \cdot M_B + L_A^{2k} \cdot L_B^{2tk}
        \enspace.
    \end{align*}
    Lines 5-13 try to compute $C = M_A \cdot M_B$.
    Then, in line 14, we sum up the 4 terms.
    Using the fact that multiplication of matrices of rank $k$ takes $O(kn^2)$ time, the last three terms can be computed in $O(tkn^2)$ time, and hence, it remains to compute $M_A \cdot M_B$. The remainder of this subsection is dedicated to analyzing lines 5-13, which contain the most involved part of the algorithm.

    We would like to compute $M_A \cdot M_B$ by writing
    $M_A = R_1 + R_2 + \dots + R_t$,
    and $M_B = S^{(r)}_1 + S^{(r)}_2 + \dots + S^{(r)}_t$ for $r = 1\dots t$, and then computing $\ALG(R_r, S^{(r)}_s)$ for all $r,s$. Note that if we could guarantee that
    $\ALG(R_r, S^{(r)}_s)$ returns $R_r \cdot S^{(r)}_s$, then, we would have
    \begin{equation}\label{eq:MA-MB-dcompose}
        M_A \cdot M_B
        = \sum_{r,s} R_r \cdot S^{(r)}_s
        = \sum_{r,s} \ALG(R_r, S^{(r)}_s)
        \enspace.
    \end{equation}
    However, $\ALG$ is not guaranteed to return the product of the inputs correctly. Instead, we claim that (1) for \emph{some} good coordinates $(i^*,j^*)$ it holds $C_{i^*,j^*} = (M_A \cdot M_B)_{i^*,j^*}$, and (2) the remaining coordinates in $C$ remain $\bot$.
    This is summarized formally in the next two claims.

        \begin{claim}\label{claim:never-wrong}
        For any $(i,j) \in [n] \times [n]$ if $C_{i,j} \in \{0,1\}$ (i.e., $C_{i,j} \neq \bot$),
        then $C_{i,j} = (M_A \cdot M_B)_{i,j}$.
    \end{claim}

    \begin{proof}
        Fix any coordinate $(i,j)$.
        Note that in line 13 we set the value of $C_{i,j} = \sum_{r,s} \ALG(R_r, S^{(r)}_s)_{i,j} \pmod 2$
        only if $\ALG(R_r, S^{(r)}_s)_{i,j} = 1$ for all $r,s$.
        Recall, that by the assumption of the algorithm if $\ALG(R_r, S^{(r)}_s)_{i,j} = 1$,
        then $\ALG(R_r, S^{(r)}_s)_{i,j} = (R_r \cdot S^{(r)}_s)_{i,j}$.
        The claim follows by \cref{eq:MA-MB-dcompose} restricted to the coordinate $(i,j)$, as
        \begin{equation*}
            C_{i,j} = \sum_{r,s} \ALG(R_r, S^{(r)}_s)_{i,j}
            = \sum_{r,s} (R_r \cdot S^{(r)}_s)_{i,j}
            = (M_A \cdot M_B)_{i,j}
            \enspace,
        \end{equation*}
        as required.
    \end{proof}

    \begin{claim}\label{claim:one-iteration}
        Fix a good coordinate $(i^*,j^*) \in G$. Then
        $\Pr[C_{i^*,j^*} = (M_A \cdot M_B)_{i^*,j^*}] \geq \delta_0 = 0.5^{O(\log^3(1/\delta))}$.
    \end{claim}

    \begin{proof}
        Consider the set $X_{i^*,j^*}$ from \cref{defgood-matrices} for a good entry $(i^*,j^*) \in G$.
        By \cref{claim-density}, the density of $X_{i,j}$ is at least ${\delta}/4$, and hence, \cref{lemma:bogolyubov-quasipoly} guarantees the existence of an affine subspace $V_{i^*,j^*}$ of dimension $\dim(V_{i^*,j^*}) \geq n-k$. Then, using \cref{lemma:extended-matrix} with $V_{i^*,j^*}$ we have
    \begin{equation}\label{eq:pr-term1}
        \Pr[M_A \in {V_{X_{i,j}}}] \geq \frac{1}{2 \cdot 2^k}
        \enspace.
	\end{equation}
	Let us condition on the event that $M^A \in {V_{X_{i^*,j^*}}}$. Then by \cref{lemma:bogolyubov-quasipoly},
	\begin{equation}\label{eq:pr-term2}
		\Pr_{\substack{R_1,R_2,..,R_t \in \F_n^{n \times n} \\ \sum_r R_r = M_A}}[R_1,R_2,..,R_t \in {X_{i,j}}] \geq (\delta/8)^t
        \enspace.
	\end{equation}
	For each of $R_1,\dots,R_t$ define the sets $Y^{R_1}, Y^{R_2},\dots,Y^{R_t}$ as in \cref{defgood-matrices}.
    (Recall $Y^R$ is the set of are all matrices $S$ such that $(R \cdot S)_{i^*,j^*} = 1$. We omit the subscript $(i^*,j^*)$ for readability.)
    
    From \cref{claim-density}, we know each of $Y^{R_1},\dots,Y^{R_t}$ has density at least $\delta/4$.
    Hence, by applying \cref{lemma:bogolyubov-quasipoly} on each of them, we obtain subspaces $V_{Y^{R_1}},..,V_{Y^{R_t}}$ of co-dimension at most $k$. Define $V_{Y} = V_{Y^{R_1}} \cap V_{Y^{R_2}} \cap \dots \cap V_{Y^{R_t}}$ to be their intersection, and note that
    $\dim(V_Y) \geq n - tk$. Therefore, by applying \cref{lemma:extended-matrix} on the matrix $B$ with the subspace $V_Y$, we get~\footnote{Note that although the algorithm samples $M_B$ before  $R_1,\dots,R_t$, in fact they are sampled independently of each other, and hence \cref{lemma:extended-matrix} is applicable here.}
	\begin{equation}\label{eq:pr-term3}
		\Pr[M_B \in V_Y] \geq \frac{1}{2\cdot{2}^{tk}}
        \enspace.
	\end{equation}
	Conditioning further on the event that $M_B \in V_Y$, we apply \cref{lemma:bogolyubov-quasipoly}, and for each $r = 1,\dots,t$ we get
	\begin{equation*}
		\Pr_{\substack{S^{(r)}_1,\dots,S^{(r)}_{t} \\ \sum_s S^{(r)}_s = M_B}}[S^{(r)}_1,\dots,S^{(r)}_t \in {Y^{R_r}}] \geq (\delta/8)^t
        \enspace.
	\end{equation*}
	Since the events above are independent between different $r$'s, the probability that the algorithm returns correct output on the entry $(i^*,j^*)$ is lower bounded by the product of the probabilities in \cref{eq:pr-term1,eq:pr-term2,eq:pr-term3}, and hence
    \begin{equation*}
        \Pr[C_{i^*,j^*} = (M_A \cdot M_B)_{i^*,j^*}]
        \geq
        \frac{1}{2^{k+1}}
        \times
        (\delta/8)^t
        \times
        \frac{1}{2\cdot{2}^{tk}}
        \times
        \left((\delta/8)^t\right)^t
        \geq \frac{1}{2^{O(\log^3(1/\delta))}}
        \enspace.
    \end{equation*}
    This completes the proof of \cref{claim:one-iteration}.

    \end{proof}

\subsection{Proof of \cref{thm:mm-alg-reduction-one-sided}}\label{sec:proof-one-sided}

We are now ready to prove \cref{thm:mm-alg-reduction-one-sided}.
\cref{alg:onesidedagreement-mm} uses \cref{alg:forgoodentries} as a subroutine, by running it several times. We claim that \cref{alg:onesidedagreement-mm} correctly computes the correct answer with high probability for any input $A,B$. Note that \cref{alg:forgoodentries} is guaranteed to be correct only for good coordinates, although it does not get the good coordinates as an input, and the guarantee about the good coordinates only appears in the analysis.

\begin{algorithm}[H]
	\caption{Approximation for one-sided Agreement Matrix Multiplication Algorithms}\label{alg:onesidedagreement-mm}
            \KwIn{$A, B \in\F_2^{n \times n}$}
            \KwOut{$A \cdot B$}

        Let $C$ be the $n \times n$ matrix initialized with all $\bot$.

        Let $\delta_0$ be the constant from \cref{claim:one-iteration}

        \RepTimes{$O\left(\frac{\log(n)}{\delta \cdot \delta_0} \right)$}
        {
            Sample uniformly random $\pi, \sigma \in [n]$.
        
            Run \cref{alg:forgoodentries} with the inputs as $A^{\pi,0},B^{0,\sigma},\ALG$. 
            
            Let $C^*$ be the resulting matrix
        
            \For{$(i,j) \in [n] \times [n]$}
            {
                \If{$C^*_{i+\pi \pmod n,j+\sigma \pmod n} \neq \bot$}
                    {Set $C_{i,j} = C^*_{i+\pi,j+\sigma}$}
            }
        }
        \Return C
    \end{algorithm}

    The following claim completes the proof of \cref{thm:mm-alg-reduction-one-sided}.
    
    \begin{claim}\label{claim:one-coord-one-sided}
        Fix a coordinate $(i,j) \in [n] \times [n]$.
        \cref{alg:onesidedagreement-mm} returns the matrix $C$ such that $\Pr[C_{i,j} = (A \cdot B)_{i,j}] \geq 1 - 1/n^3$.

        In particular, by taking the union bound over all coordinates $(i,j)$ it follows that for any input $A,B$ \cref{alg:onesidedagreement-mm} returns their product with probability at least $1-1/n$. 
    \end{claim}

    \begin{proof}
        Fix a coordinate $(i,j) \in [n] \times [n]$.
        The algorithm chooses random $\pi$ and $\sigma$, and runs \cref{alg:forgoodentries} on the shifted matrices $A^{\pi,0}$ and $B^{0,\sigma}$.

        Note that since $\pi$ and $\sigma$ are chosen uniformly at random, it follows that $\Pr[(i+\pi \pmod n,j+\sigma \pmod n) \in G] = \abs{G}/n^2 \geq \delta/2$.

        Suppose that $(i+\pi \pmod n,j+\sigma \pmod n)$ is indeed a good coordinate. Then by \cref{claim:one-coord-one-sided} with probability $\delta_0$ we obtain the correct answer in the coordinate $(i+\pi \pmod n,j+\sigma \pmod n)$, in which case we set $C_{i,j}$ to be that answer $(A \cdot B)_{i,j}$.
        Otherwise, $C_{i,j}$ remains $\bot$.

        Therefore, with probability at least $(\delta/2) \cdot \delta_0$ in each iteration $C_{i,j}$ changes from $\bot$ to $(A \cdot B)_{i,j}$, and once it changes, it never changes its value again.

        By repeating the procedure $R = O(\frac{\log(n)}{\delta \cdot \delta_0})$ times, the probability that in the end of the algorithm $C_{i,j} = \bot$ is upper bounded by
        $\Pr[C_{i,j} = \bot] \leq (1-\delta_0)^{R} < 1/n^3$. This completes the proof of the claim.
    \end{proof}

We conclude the proof with the analysis of the running time of the algorithm.

\paragraph{Running Time:}
The total running time of \cref{alg:onesidedagreement-mm} is essentially dominated by the running time of \cref{alg:forgoodentries} multiplied by $O\left(\frac{\log(n)}{\delta \cdot \delta_0} \right)$. Each iteration of \cref{alg:forgoodentries} involves $O(t^2)$ calls to $\ALG$ plus additional $O(tn^2)$ time. Therefore, the running time is $O(t^2 \log(n) T(n)/\delta \cdot \delta_0)$. Since $t=O(\log 1/{\delta})$, and $\delta_0 = 0.5^{O(\log^3(1/\delta))}$ the total running time of the algorithm is is $2^{O(\log^3(1/\delta))} \cdot T(n)\log(n)$.

In particular, even for a slightly sub-constant $\delta \geq \exp(-\log^{0.33}(n))$, our algorithm runs in time $T(n) \cdot n^{o(1)}$.

\medskip

This completes the proof of \cref{thm:mm-alg-reduction-one-sided}.

\bibliographystyle{alpha}
\bibliography{refs}
\end{document}